\documentclass[11pt,a4paper,reqno]{amsart}
\usepackage{fullpage} 
\usepackage{amsfonts}
\usepackage{amsthm}
\usepackage[utf8]{inputenc}
\usepackage[T1]{fontenc}
\usepackage{amsmath,bm} 
\usepackage{graphicx}
\usepackage{enumitem}
\usepackage{wrapfig}
\usepackage{dsfont}
\usepackage{braket}
\usepackage{amsxtra}
\usepackage{hyperref}

\newtheorem{de}{Definition}
\newtheorem{theo}[de]{Theorem}    

\newtheorem{lem}[de]{Lemma}

\newcommand{\Inte}[1][3]{ \int_{\mathbb{R}^{#1}} }
\newcommand{\add}{a_{\textrm{dd}}}
\newcommand{\m}{\textbf{m}}
\newcommand{\Cs}{C_{\textrm{Sob}}}
\newcommand{\nn}{\nonumber}
\newcommand{\dd}{\textrm{d}}

\title{Existence of minimizers in generalized Gross-Pitaevskii theory with the Lee-Huang-Yang correction}
\date{}
\author{Arnaud Triay}
\begin{document}
\begin{abstract}
	We study the dipolar Gross-Piteavskii functional with the Lee-Huang-Yang (LHY) correction term without trapping potential and in the regime where the dipole-dipole interaction dominates the repulsive short-range interaction. We show that, above a critical mass, the functional admits minimizers and we prove their regularity and exponential decay. We also estimate the critical mass in terms of the parameters of the system.
\end{abstract}

\maketitle
\setcounter{tocdepth}{1}
\tableofcontents
\section{Introduction}

Recent progress in the realization of Bose-Einstein condensates (BEC) of \emph{dipolar} atoms have exhibited new and complex phenomena as compared with simpler chemical elements \cite{Kadau-16,Ferrier-16,Chomaz-16,Schmitt-16,Chomaz-18}. For a survey on the properties of dipolar BEC, see \cite{LahMenSanLewPfau-09}. This specificity is due to the long-range nature of the dipolar interaction, which persists in the dilute regime, as opposed to short-range ones which become delta like. The resulting model thus contains two competing interactions, one partly attractive and anisotropic, another one repulsive. The system is stable when the repulsive term compensates the attractiveness of the dipole-dipole interaction (DDI), or when the trapping potential is strongly confining in some well-chosen directions. Otherwise, the condensate experiences collapse and blow up, a phenomenon known as Bose-nova \cite{Santos-00,KocLahMet-05,Lahaye-08}. Nevertheless, new experiments have revealed that when the scattering length, which controls the scattering length of the repulsive interaction, is lowered sufficiently slowly, the condensate remains in a meta-stable state leading to the formation of stable self-bounded droplets \cite{Kadau-16,Ferrier-16,Chomaz-16,Schmitt-16}. This phenomenon has been investigated theoretically \cite{Lima-11,Lima-12} and numerically \cite{Bisset-16,Baillie-16,Wachtler-16}. The stabilization mechanism is believed to be caused by the Lee-Huang-Yang corrections \cite{LeeHuangYang-57,Petrov-15} and is accounted for in the Gross-Pitaevskii framework by a term proportional to $|\psi|^5$ where $\psi$ is the wave function of the condensate. Namely, the generalized Gross-Pitaevskii functional used in the physics literature \cite{Wachtler-16,Baillie-16,Schmitt-16} is given by the following
\begin{equation}
	\label{eq_GP_0}
	\mathcal{E}^{GP}_{\textrm{dip}}(\psi) = \Inte |\nabla \psi|^2 + \frac{a_s}{2} \Inte |\psi|^4 + \frac{\add}{2} \Inte K\star |\psi|^2 |\psi|^2 + \frac{2}{5}\gamma_{\textrm{QF}} \Inte |\psi|^5,
\end{equation}
where $\psi \in H^1(\mathbb{R}^3)$ is normalized as
\begin{equation}
	\label{eq:def_K_dip}
\int_{\mathbb{R}^3} |\psi|^2 = N,
\end{equation}
with $N$ the number of atoms.
The parameter $a_s$ is proportional to the scattering length, the parameter $\add$ is proportional to the square of the moment of the dipoles and we will always assume $\add > 0$. Finally, the coefficient $\gamma_{\textrm{QF}} > 0$ in front of the LHY corrections physically depends on $a_s$ and $\add$ but we take it independent of them for the analysis. The acronym QF stands for quantum fluctuation which is the term used in the physics literature to refer to the cause of the LHY corrections. For a dipolar Bose gas, we have $K = K_{\textrm{dip}}$ with
\begin{equation*}
%	\label{eq:def_K_dip}
K_{\textrm{dip}}(x) = \frac{3}{4\pi}\frac{1-3\cos^2(\theta_x)}{|x|^3} =: \frac{\Omega_{\textrm{dip}}(x/|x|)}{|x|^3},
\end{equation*}
where $\cos(\theta_x) = n\cdot x / |x|$ and where $n$ is a fixed unit vector aligned with all the dipoles. Here we consider a general long-range interaction of the fom
\begin{equation}
	\label{eq:def_K}
K(x) = \frac{\Omega(x/|x|)}{|x|^3} \quad \textrm{ with } \int_{\mathbb{S}^2} \Omega(y) \dd \sigma (y) = 0,
\end{equation}
where $\Omega$ is an even continuous function on the sphere $\mathbb{S}^2$ and $\sigma$ is the Haar measure on the sphere.

Solitary waves are solutions of  the following Gross-Pitaevskii equation
\begin{equation}
	\label{eq:EL_intro}
-\Delta \psi + a_s |\psi|^2 \psi + \add K \ast |\psi|^2 \psi + \gamma_{\textrm{QF}} |\psi|^3 \psi = \mu \psi,
\end{equation}
for some chemical potential $\mu \in \mathbb{R}$. They can be obtained by looking at critical points of $\mathcal{E}^{GP}_{\textrm{dip}}$ restricted to the unit  sphere in $L^2(\mathbb{R}^3)$. Of course, the easiest way to find critical points is to minimize the functional and look for the ground state. 

Without the LHY correction ($\gamma_{\textrm{QF}} =0$) and with a trapping potential $V_{\textrm{ext}}$, the functional $\mathcal{E}^{GP}_{\textrm{dip}}$ has been extensively studied \cite{BaoCaiWan-10,BaoAbdCai-12,CarMarkSpa-08,CarHaj-15,CarMarkSpa-08}. A necessary and sufficient condition for $\mathcal{E}^{GP}_{\textrm{dip}}$ to be bounded below on the unit sphere of $L^2(\mathbb{R}^3)$ is, in our units, $a_s \geq \add$. This comes from the fact that the kinetic energy term and the interaction term do not have the same scaling properties with respect to dilatations. Adding the (positive) LHY term allows to prevent the collapse and to access the previously unstable regime $a_s < \add$. Without confining potential, a necessary condition for the minimizing sequences of $\mathcal{E}^{GP}_{\textrm{dip}}$ to be pre-compact is the negativity of the ground state energy, as we will prove. But because of the stabilization mechanism itself, this condition does not hold for all possible choice of parameters $a_s, \add, \gamma_{\textrm{QF}}$ and $N$. The case $\gamma_{\rm QF} < 0$ was analyzed in \cite{LuoAth-18} where the authors find solutions to (\ref{eq:EL_intro}) by means of mountain pass arguments.

In this paper, we study the existence and non existence of the minimizers of (\ref{eq_GP_0}) as well as their regularity. Our main result is Theorem \ref{theo_min} in which we show that the minimum energy is decreasing in $N$ and that there is some critical mass $N_c(a_s,\add,\gamma_{\textrm{QF}})$ below which it is zero and there is no ground state and above which it is negative and there is at least one ground state. We also derive some upper and lower bound on $N_c$.

During the preparation of this work a similar result was announced \cite{LuoAth-19} in which the authors study the existence of standing waves for the dipolar Gross-Pitaevskii functional with the LHY non-linearity replaced by $|\psi|^p$ for $p\in (4,6]$. However, the case of a general long-range interaction given by (\ref{eq:def_K}) is not dealt with and does not seem to follow from their proof. Their approach uses the particular symmetry of the dipole-dipole potential which allows to reformulate the interaction energy as the sum of a local term and another term involving the Riesz transform.

\subsubsection*{\textbf{Acknowledgment}}
This project has received funding from the European Research Council (ERC) under the European Union's Horizon 2020 research and innovation programme (grant agreement MDFT No 725528 of Mathieu Lewin).

\section{Main results}

We first rescale the functional to get rid of redundant parameters. For $\lambda > 0$ and $\psi \in H^1(\mathbb{R}^3)$ such that $\int_{\mathbb{R}^3} |\psi|^2 = \lambda $, we denote $\psi_{\alpha,\ell} = \alpha^{1/2} \ell^{3/2} \psi(\ell \cdot)$ and compute
\begin{align*}
\mathcal{E}^{GP}_{\textrm{dip}}(\psi_{\alpha,\ell}) = \alpha \ell^{2} \left( \Inte |\nabla \psi|^2 + \frac{a_s\alpha \ell }{2} \Inte |\psi|^4 + \frac{\add \alpha \ell }{2} \Inte K\star |\psi|^2 |\psi|^2 + \frac{2}{5}\gamma_{\textrm{QF}} \alpha^{3/2} \ell^{5/2}  \Inte |\psi|^5 \right).
\end{align*}
Note that $N = \alpha \lambda$. Taking $\alpha \ell a_s = 1$ and $\alpha^{3/2} \ell^{5/2} \gamma_{QF} = 1$, denoting $b = \add / a_s$ and dividing by $\alpha \ell^2$ we obtain
\begin{equation}
	\label{eq_GP}
	\mathcal{E}_b(\psi) := \Inte |\nabla \psi|^2 + \frac{1}{2} \Inte |\psi|^4 + \frac{b}{2} \Inte K\star |\psi|^2 |\psi|^2 + \frac{2}{5} \Inte |\psi|^5,
\end{equation}
with the new constraint $\int_{\mathbb{R}^3} |\psi|^2 = \lambda$.

Recall that the third term with $K$ has to be understood in the sense of the principal value, that is $K = \lim_{\varepsilon \to 0} \mathds{1}_{|x|>\varepsilon} K$ in $\mathcal{D}'$. It is classical \cite{Duo01} that when $\Omega$ is an even continuous function on the sphere $\mathbb{S}^2$ satisfying (\ref{eq:def_K}), then for $1<p<\infty$ and any $f\in L^p(\mathbb{R}^3)$, the following limit exists for almost all $x\in \mathbb{R}^3$
\begin{equation*}
K\ast f(x) := \lim_{\varepsilon \to 0} (\mathds{1}_{|x|>\varepsilon} K) \ast f (x).
\end{equation*}
Moreover there exists some constant $C_p>0$ such that for any $\varepsilon >0$ and any $f\in L^p(\mathbb{R}^3)$ we have
\begin{align}
	\label{eq:lp_continuity_espilon}
\| (\mathds{1}_{|x|>\varepsilon} K) \ast f \|_{L^p(\mathbb{R}^3)} & \leq C_p \|f \|_{L^p(\mathbb{R}^3)}, \\
	\label{eq:lp_continuity}
\|K \ast f \|_{L^p(\mathbb{R}^3)}  &\leq C_p \|f \|_{L^p(\mathbb{R}^3)}.
\end{align}
Additionally, $\widehat{K} \in L^\infty(\mathbb{R}^3)$ and up to modifying the parameter $b$, we will assume that $$\inf \widehat{K} = -1.$$

We now state our main result. We first need some notations and definitions. For $\lambda,b >0$, we denote by 
\begin{equation}
	\label{eq_def_min}
E(\lambda,b) = \inf_{ \int |\psi|^2 = \lambda } \mathcal{E}_b(\psi)
\end{equation}
the ground state energy with mass contraint $\lambda$.

\begin{theo}
	\label{theo_min}
For any fixed $b>0$, the function 
\begin{align*}
\lambda \in \mathbb{R}_{+} \mapsto E(\lambda,b)
\end{align*}
is non-increasing and there exists $0< \lambda_c(b) < \infty$ such that the following hold
\begin{itemize}
\item
for $ 0 < \lambda < \lambda_c(b)$, we have $E(\lambda,b) = 0$ and there is no minimizer

\item
the function $\lambda \mapsto E(\lambda, b)$ is strictly decreasing on $[\lambda_c(b),+\infty)$. For every $\lambda \geq \lambda_c(b)$ it admits at least one minimizer. It solves the equation
\begin{equation}
	\label{eq:EL}
\left(-\Delta + |\psi|^2 + b K\star|\psi|^2 + |\psi|^3 - \mu \right) \psi = 0,
\end{equation}
with $\mu <0$. In addition, $\psi$ is $C^\infty$ and decays exponentially.
\end{itemize}
Finally we have
\begin{equation}
	\label{eq:domain_admissible_par}
\frac{2^{1/2}5^{1/2} 3 \pi}{(b-1)^{5/2}}\leq \lambda_c(b),
\end{equation}
and in the dipolar case (\ref{eq:def_K_dip}) we have
\begin{equation}
		\label{eq:domain_admissible_par_upper_bound}
\lambda_c(b) \leq 84.437 \frac{1}{(b-1)^{5/2}}.
\end{equation}
\end{theo}

It is an interesting problem to derive the exact asymptotics of $\lambda_c(b)$ as $b\to1$ or as $b\to \infty$. The rest of the paper is dedicated to the proof of Theorem \ref{theo_min}.

\section{Proof of Theorem \hyperref[theo_min]{1}}
When $b$ is fixed, to simplify notations, we will simply denote by $E(\lambda)$ and $\mathcal{E}$ respectively $E(\lambda,b)$ and $\mathcal{E}_b$.
\subsection{Monotonicity of the $\lambda \mapsto E(\lambda,b)$}
Let $b>0$ and show that $\lambda \mapsto E(\lambda)$ is non-increasing. It suffices to prove that for all $\lambda_1, \lambda_2 \geq 0$,
\begin{equation}
	\label{eq:subadditive}
E(\lambda_1 + \lambda_2) \leq E(\lambda_1) + E(\lambda_2).
\end{equation}
and that $E \leq 0$. By density we can take $\psi^{(i)}$, for $i\in\{1,2\}$, with compact support and such that $\mathcal{E}(\psi^{(i)}) \leq E(\lambda_i) + \varepsilon$, for some $\varepsilon >0$. We then obtain
\begin{align*}
E(\lambda_1 + \lambda_2) 
	&\leq \lim_{t \to \infty}\mathcal{E}(\psi^{(1)} + \psi^{(2)}(\cdot - t e_1)) = \mathcal{E}(\psi^{(1)}) + \mathcal{E}(\psi^{(2)}) \\
	&\leq E(\lambda_1) + E(\lambda_2) + 2 \varepsilon,
\end{align*}
where $e_1 = (1,0,0)$. It remains to take $\varepsilon$ to zero to obtain (\ref{eq:subadditive}). To show that $E(\lambda)\leq 0 $ for all $\lambda \geq 0$, we take $\varphi \in C^\infty(\mathbb{R}^3)$ with compact support and such that  $\int_{\mathbb{R}^3}|\varphi|^2 = \lambda$, we denote $\varphi_n (x) = n^{-3/2} \varphi(n^{-1}(x-n^2))$ for $n\geq 1$. We have $\int_{\mathbb{R}^3} |\nabla\varphi_n|^2 = n^{-2} \int_{\mathbb{R}^3} |\nabla\varphi|^2 $ and $\|\varphi_n\|_{L^p(\mathbb{R}^3)} = n^{3(1/p - 1/2)} \|\varphi\|_{L^p(\mathbb{R}^3)}$ for $p \geq 1$. Using (\ref{eq:lp_continuity}), we obtain $\mathcal{E}(\varphi_n) \to 0$ as $n \to \infty$.

\subsection{ $\bm{0 < \lambda_c(b) < \infty}$}

To prove that $\lambda_c(b) >0$ it suffices to show the lower bound (\ref{eq:domain_admissible_par}). Let $\psi \in H^{1}(\mathbb{R}^3)$. By Hölder's inequality we have 
\begin{align*}
\|\psi\|_{L^6(\mathbb{R}^3)}^2 &\geq \|\psi\|_{L^2(\mathbb{R}^3)}^{-2/9} \|\psi\|_{L^5(\mathbb{R}^3)}^{20/9}, \\
- \|\psi\|_{L^4(\mathbb{R}^3)}^4 &\geq - \|\psi\|_{L^2(\mathbb{R}^3)}^{2/3} \|\psi\|_{L^5(\mathbb{R}^3)}^{10/3}.
\end{align*}
On the other hand
\begin{equation*}
\int_{\mathbb{R}^3} K \ast |\psi|^2 |\psi|^2 = \int_{\mathbb{R}^3} \widehat{K} \left|\widehat{|\psi|^2}\right|^2 \geq - \int_{\mathbb{R}^3} |\psi|^4.
\end{equation*}
since we have, by assumption, $\inf \widehat{K} = 1$.
Denoting $\lambda = \|\psi\|_{L^2(\mathbb{R}^3)}^2$ and $X = \|\psi\|_{L^5(\mathbb{R}^3)}^5$, we can then bound by below the Gross-Pitaevskii energy in the following way
\begin{align*}
\mathcal{E}(\psi) \geq \Cs \lambda^{-1/9} X^{4/9} - \frac{b-1}{2} \lambda^{1/3} X^{2/3} + \frac{2}{5} X =: F_{1}(\lambda,X),
\end{align*}
where $\Cs = 3 (2\pi)^{2/3}/4$ is the optimal constant in Sobolev's inequality.
% and $C_K = (\add - a_s)/2 = a_s(\varepsilon - 1)/2$ with $\varepsilon = \add/a_s$.
We want to compute
\begin{equation*}
\lambda_0 = \sup \left\{ \lambda > 0 \,  \big| \,  F_{1}(\lambda,X) \geq 0, \forall X \geq 0 \right\}.
\end{equation*}
By the form of $F_{1}$, there is a unique such $\lambda_0$, and there is some $X_0 >0$ such that the following system holds
\begin{equation*}
\left\{
\begin{array}{ll}
\quad \quad F_{1}(\lambda_0,X_0) &= 0, \\
(\partial_X F_{1})(\lambda_0,X_0) &= 0.
\end{array}
\right.
\end{equation*}
Solving this, we obtain
\begin{align*}
\lambda_0 =  \left(\frac{5}{3}\right)^{1/2} 2^{5/2} \frac{\Cs^{3/2}}{(b-1)^{5/2}} = \frac{2^{1/2}5^{1/2} 3 \pi}{(b-1)^{5/2}}.
%& \simeq 29.8037  \frac{1}{(b-1)^{5/2}}.
\end{align*}

To prove that $\lambda_c(b) < \infty$, let us take $\psi \in H^1(\mathbb{R}^3)$ such that $ \int_{\mathbb{R}^3} |\psi|^2 = 1$ and
\begin{equation*}
\int_{\mathbb{R}^3} |\psi|^4 + b K \ast |\psi|^2 |\psi|^2 < 0
\end{equation*}
which is possible because $b >1$, see the proof of \cite[Theorem 1]{Tri-18}. With $\psi_{\lambda,\ell} = \lambda^{1/2} \ell^{3/2} \psi(\ell \cdot)$, we have
\begin{align*}
\mathcal{E}(\psi_{\lambda,\ell}) = \lambda \ell^{2} \left( \Inte |\nabla \psi|^2 + \frac{\lambda \ell }{2} \left(\Inte |\psi|^4 + \frac{b}{2} \Inte K\star |\psi|^2 |\psi|^2\right) + \frac{2}{5} \alpha^{3/2} \ell^{5/2}  \Inte |\psi|^5 \right),
\end{align*}
taking $\lambda \to \infty$ and $\ell = \lambda^{-1/2}$, we obtain $\mathcal{E}(\psi_{\lambda,\ell}) \to - \infty$, which proves that $\lambda_c(b) < \infty$ since $E(\lambda)$ cannot stay equal to $0$ on $\mathbb{R}_+$.

\subsection{Properties of solutions to (\ref{eq:EL})}
If $\psi$ is a stationary point of the functional $\mathcal{E}$ under the mass constraint $\lambda$, it solves the Euler-Lagrange equation (\ref{eq:EL}) for some Lagrange multiplier $\mu$. We begin by showing some relations between $\mu, \lambda$ and the different terms in $\mathcal{E}(\psi)$. We define
\begin{align*}
T(\psi) &= \int |\nabla\psi|^2 >0, \\
I(\psi) &= \frac{a_s}{2} \int |\psi|^4 + \frac{\add}{2} \int K\star |\psi|^2 |\psi|^2, \\
Q(\psi) &= \frac{2}{5}\gamma \int |\psi|^5 >0.
\end{align*}

When there is no ambiguity, we will simply denote them by $T,I$ and $Q$.

\begin{lem}
	\label{lem_virial}
Let $\lambda,b >0$ and let $\psi$ be solution to (\ref{eq:EL}) with $\int_{\mathbb{R}^3} |\psi|^2 = \lambda$. Let $\mu$ be the associated Lagrange multiplier in (\ref{eq:EL}), then following equalities hold:
\begin{align}
T &= \frac{3}{2}\left(E - \mu \lambda\right), \label{eq_vir_1}\\
I &= \frac{E + 5\lambda \mu}{2},  \\
Q &= - (\mu \lambda + E).
\end{align}
In particular, if $E\leq 0$ then $\mu < 0$.
\end{lem}
\begin{proof}
We have
\begin{align}
T + I + Q &= E(\lambda) \label{eq_pf_vir_1}, \\
I + \frac{3}{2} Q &= -E(\lambda) +\mu\lambda \label{eq_pf_vir_2},\\
2 T + 3 I + \frac{9}{2} Q &= 0.\label{eq_pf_vir_3}
\end{align}
The equation (\ref{eq_pf_vir_1}) is simply the definition of $\mathcal{E}$. Equation (\ref{eq_pf_vir_2}) is obtained by integrating the Euler-Lagrange equation (\ref{eq:EL}) against the solution $\psi$. Finally, (\ref{eq_pf_vir_3}) is a consequence of the virial theorem: for $\alpha > 0$ we denote $\psi_\alpha = \alpha^{3/2}\psi(\alpha \cdot )$ and since the function $\alpha \mapsto \mathcal{E}(\psi_\alpha)$ is stationary at $\alpha = 1$, its derivative vanishes, this gives (\ref{eq_pf_vir_3}) (we can also obtain (\ref{eq_pf_vir_3}) by multiplying (\ref{eq:EL}) by $x \cdot \nabla \psi$). It remains to solve the linear system (\ref{eq_pf_vir_1}),(\ref{eq_pf_vir_2}),(\ref{eq_pf_vir_3}) which gives the result. The negativity of $\mu$ comes from (\ref{eq_vir_1}), indeed we have $\lambda \mu = E - 2T/3 < 0$, if $E\leq 0$.
\end{proof}

We then state some lemma about the regularity of the solutions of the Euler-Lagrange equation (\ref{eq:EL}). We give its proof in Appendix \ref{app:Lemma_regularity}.

\begin{lem}\label{lemma_reg_minimizers}
Let $\psi \in H^1(\mathbb{R}^3)$ be a solution of the Euler-Lagrange equation (\ref{eq:EL}) for some $\mu < 0$. Then $\psi \in C^{\infty}(\mathbb{R}^{3})$ and there exists some constants $C,t>0$ such that for all $x\in\mathbb{R}^{3}$
\begin{equation*}
0< |\psi(x)| \leq C e^{-t|x|}.
\end{equation*}
\end{lem}

\subsection{Existence of a minimizer, case $\lambda > \lambda_c$}

We begin with the case $\lambda > \lambda_c(b)$. We use the concentration-compactness method \cite{Lions-84,Lions-84b,Lieb-83}. In particular, we follow \cite{Lewin-VMQM} and use the following lemma.
\begin{lem}\label{lem_rel_compact}
Let $\{\psi_n\}$ be any bounded sequence in $H^1(\mathbb{R}^3)$, define 
\begin{equation*}
\m(\{\psi_n\}) = \sup \left\{ \int |\psi|^2 \quad \Big| \quad \exists \{x_{n_k}\} \subset \mathbb{R}^3, \psi_{n_k}(\cdot - x_{n_k}) \rightharpoonup \psi \textrm{ weakly in } H^1(\mathbb{R}^3) \right\}.
\end{equation*}
Then the following assertions are equivalent
\begin{itemize}
\item $\m(\{\psi_n\}) = 0$
\item $\psi_n \to 0$ strongly in $L^p(\mathbb{R}^3)$ for all $2<p<6$.
\end{itemize}
\end{lem}

We can now show that minimizing sequences are precompact in $H^1(\mathbb{R}^3)$. Let $\{\psi_n\} \subset H^1(\mathbb{R}^3)$, $\int_{\mathbb{R}^3} |\psi_n|^2 = \lambda$, with $\lambda > \lambda_c(b)$ so that $E(\lambda) < 0$, be a minimizing sequence for $\mathcal{E}$. By (\ref{eq:lp_continuity}) we verify that $\{\psi_n\}$ is bounded in $H^1(\mathbb{R}^3)$ so that we can apply the above lemma. If $\m(\{\psi_n\}) = 0$ then it follows that $\mathcal{E}(\psi_n) \to 0$ which contradicts $E(\lambda) <0$. We can therefore find $Q_1 \in H^1(\mathbb{R}^3)$, $Q_1 \not\equiv 0$, such that, up to translations, $\psi_n \rightharpoonup Q_1$ weakly in $H^1(\mathbb{R}^3)$. We denote
\begin{equation*}
\psi_n = Q_1 + r_n, \qquad r_n \rightharpoonup 0 \quad \textrm{ in } \quad  H^1(\mathbb{R}^3).
\end{equation*}
Because of the weak convergence in $H^1(\mathbb{R}^3)$, we have
\begin{equation*}
\int_{\mathbb{R}^3} |\nabla \psi_n|^2 = \int_{\mathbb{R}^3} |\nabla Q_1|^2 + \int_{\mathbb{R}^3} |\nabla r_n|^2 + o(1).
\end{equation*}
Moreover, up to extracting a subsequence, we can assume strong local convergence in $L^2(\mathbb{R}^3)$ and convergence. By \cite[Theorem 1.9]{LieLos-01}, we obtain
\begin{equation*}
\lim_{n\to\infty} \int_{\mathbb{R}^3} |\psi_n|^p - |r_n|^p = \int_{\mathbb{R}^3} |Q_1|^p,
\end{equation*}
for all $2\leq p \leq 6$. To deal with the non-local term, we note that $\|r_n Q_1\|_{L^2(\mathbb{R}^3)} \to 0$ as $n\to \infty$ since $r_n^2 \rightharpoonup 0$ in $L^2(\mathbb{R}^3)$. Hence by (\ref{eq:lp_continuity}) we conclude that
\begin{equation*}
\int_{\mathbb{R}^3} K \ast |\psi_n|^2 |\psi_n|^2 = \int_{\mathbb{R}^3} K \ast |Q_1|^2 |Q_1|^2 + \int_{\mathbb{R}^3} K \ast |r_n|^2 |r_n|^2 + o(1).
\end{equation*}
Denoting $\lambda_1 = \|Q_1\|^2_{L^2(\mathbb{R}^3)}$ and $\widetilde{r}_n = (\lambda - \lambda_1) r_n / \| r_n\|_{L^2(\mathbb{R}^3)}$, so that $\|\widetilde{r}_n -r_n\|_{H^1(\mathbb{R}^3)} \to 0$, we obtain from the previous estimates
\begin{align}
\mathcal{E}(\psi_n) 
&= \mathcal{E}(Q_1) + \mathcal{E}(\widetilde{r}_n) + o(1) \label{eq:min_energy_1} \\
&\geq \mathcal{E}(Q_1) + E(\lambda-\lambda_1) + o(1), \nn
\end{align}
where we used that $\mathcal{E}$ is locally uniformly continuous in $H^1(\mathbb{R}^3)$. From this we obtain that $E(\lambda) = E(\lambda_1) + E(\lambda - \lambda_1)$, that $Q_1$ is a minimizer of $\mathcal{E}$ for the mass constraint $\lambda_1$ and that $\{\widetilde{r}_n\}$ is a minimizing sequence for the mass constraint $\lambda-\lambda_1$. If $\lambda_1 = \lambda$ the result is proved since $E(0) = 0$. Let us then assume $\lambda_1 < \lambda$, in this case one must have $\m(\{\widetilde{r}_n^{(2)}\}) > 0$, otherwise we obtain by Lemma \ref{lem_rel_compact} that $\liminf \mathcal{E}(\widetilde{r}_n) \geq 0$ and from (\ref{eq:min_energy_1}) that $E(\lambda) = E(\lambda_1)$, since $E$ is non-increasing. But $Q_1$ is a minimizer of $\mathcal{E}$ and therefore satisfies (\ref{eq:EL}) with some $\mu < 0$ by Lemma \ref{lem_virial}. Hence for any $\varepsilon >0$ small enough we have
\begin{align}
\mathcal{E}((1+\varepsilon)Q_1) 
=\mathcal{E}(Q_1) + 2 \mu \lambda_1 \varepsilon + o(\varepsilon) < \mathcal{E}(Q_1), \label{eq:derivee_negative}
\end{align}
since $\mu < 0$. For $\varepsilon >0$ small enough we thus obtain $E(\lambda) \leq E(\lambda_1 + \varepsilon) < E(\lambda_1)$ which contradicts $E(\lambda) = E(\lambda_1)$, hence $\m(\{\widetilde{r}_n^{(2)}\}) > 0$. Doing the same procedure for $\widetilde{r}_n$ as we did with $\psi_n$ we obtain
\begin{equation*}
\widetilde{r}_n = Q_2 + q_n, \qquad q_n \rightharpoonup 0 \quad \textrm{ in } \quad H^1(\mathbb{R}^3),
\end{equation*}
and we denote $\lambda_2 = \int_{\mathbb{R}^3}|Q_2|^2$. From (\ref{eq:min_energy_1}), the same computations as before lead to
\begin{equation*}
E(\lambda) \leq E(\lambda_1) + E(\lambda_2) + E(\lambda-\lambda_1 -\lambda_2),
\end{equation*}
and $\mathcal{E}(Q_2) = E(\lambda_2)$. From this and (\ref{eq:subadditive}) we deduce that
\begin{equation}
	\label{eq_egalite_min_2_contraintes}
E(\lambda_1 + \lambda_2) = E(\lambda_1) + E(\lambda_2),
\end{equation}
which we will prove cannot hold. Since $(Q_i)_{i\in \{1,2\}}$ are minimizers, they satisfy the Euler-Lagrange equation (\ref{eq:EL}) and by Lemma \ref{lemma_reg_minimizers} they are $C^{\infty}$ and have exponential decay. Let us write $Q_i = Q_i \chi(\cdot / R) + Q_i (1-\chi(\cdot / R)) = Q_{(i,1)} +Q_{(i,2)}$ for some $\chi \in C^\infty_c(\mathbb{R}^3)$ with $\chi\equiv 1$ in $B(0,1)$ and $\chi\equiv 0$ in $\mathbb{R}^3\setminus B(0,2)$. We have $\|Q_{(i,2)}\|_{L^p(\mathbb{R}^3)} \leq C e^{-t R}$, for $2\leq p \leq 6$, for some constants $C,t>0$. Taking $u \in \mathbb{S}^2$, we define 
\begin{equation*}
\psi_R = Q_1 + Q_2(\cdot- R^2 u).
\end{equation*}
Since $\|\psi_R\|_{L^2(\mathbb{R}^3)} \to \lambda_1 + \lambda_2$ when $R \to \infty$ we can write $\psi_R = \widetilde{\psi}_R + \psi_R\widetilde{\psi}_R$ with $\|\widetilde{\psi}_R\|_{L^2(\mathbb{R}^3)} = \lambda_1 + \lambda_2$ and $\|\psi_R - \widetilde{\psi}_R\|_{H^1(\mathbb{R}^3)} = O(e^{-tR})$ when $R \to \infty$ in $H^1(\mathbb{R}^3)$.
We then have
\begin{align}
\mathcal{E}(\psi_R)
&=\mathcal{E}(\widetilde{\psi}_R)   + O(e^{-tR})  \nn \\
&= \mathcal{E}(Q_1) + \mathcal{E}(Q_2) + b \int_{\mathbb{R}^3} K \star Q_{1,1}^2 Q_{2,1}^2(\cdot -R^2 u ) + \mathcal{O}(e^{-tR}) +  \mathcal{O}(e^{-tR^2}),\label{eq:energy}
\end{align}
as $\alpha \to \infty$, where we used (\ref{eq:lp_continuity}). We will now use the following lemma whose proof is postponed until the end of the argument.

\begin{lem}\label{lem_dip_pot_negatif}
Let $\phi \in C^0(\mathbb{R}^3) \cap L^2(\mathbb{R}^3)$. Then for all $u\in\mathbb{S}^2$, 
\begin{equation*}
\sup_{v\in B(0,R)} K\star( |\phi|^2\mathds{1}_{B(0,R)} )(R^2 u + v) \leq \; \frac{\int_{\mathbb{R}^3} |\phi|^2 }{R^6}\left(\Omega(u) +o(1)\right),
\end{equation*}
as $R \to \infty$
\end{lem}

We use the above lemma with $u$ such that $\Omega(u) = \inf \Omega < 0$ and obtain
\begin{align*}
 \int_{\mathbb{R}^3} K \star Q_{1,1}^2 Q_{2,1}^2(\cdot -R^2 u ) \leq \frac{1}{R^6} \Inte Q_{1,1}^2 \Inte Q_{2,1}^2\left(\inf\Omega + o(1)\right).
\end{align*}
Using this in (\ref{eq:energy}) gives
\begin{align*}
E(\lambda_1 + \lambda_2) &\leq \mathcal{E}(Q_1 + Q_2(\cdot- \alpha u)) \\
&\leq  \mathcal{E}(Q_1) + \mathcal{E}(Q_2) +\lambda_1 \lambda_2 \frac{\inf\Omega}{R^6} + \mathcal{O}(e^{-tR}) \\
&\leq E(\lambda_1) + E(\lambda_2) + \lambda_1 \lambda_2\frac{\inf\Omega}{R^6} + \mathcal{O}(e^{-tR}) \\
&< E(\lambda_1) + E(\lambda_2),
\end{align*}
which contradicts (\ref{eq_egalite_min_2_contraintes}). 

We conclude that $\lambda_2 = 0$ and that there exists $\{x_n\}\subset \mathbb{R}^3$ such that $\{\psi_n(\cdot-x_n)\}$ is relatively compact in $L^p(\mathbb{R}^{3})$ for $p\in[2,6)$. We can obtain a minimizer by extracting a converging subsequence. By a classical argument we indeed have strong convergence in $H^1(\mathbb{R}^3)$.

\begin{proof}[Proof of Lemma \ref{lem_dip_pot_negatif}]
We extend the function $\Omega$ to the whole of $\mathbb{R}^3$ by $\Omega(x) = \Omega (x/|x|)$. Let $\omega$ be the modulus of uniform continuity of $\Omega$ in $B(0,2) \setminus B(0,1/2)$. For $u \in \mathbb{S}^2$ we have
\begin{align*}
\alpha^3 \Bigg|\Inte \frac{\Omega(R^2 u + v-y)}{|R^2 u + v-y|^3}& |\phi(y)|^2 dy -  \Omega(u) \; \frac{\int_{\mathbb{R}^3} |\phi|^2 }{R^6} \Bigg| \\
&\leq
\left| \int_{B(0,R)} \frac{\Omega(u + (v-y)/R^2) - \Omega(u)}{|u + (v-y)/R^2|^3} |\phi(y)|^2 dy\right| \\
& \qquad\qquad\qquad\qquad
+ \left|  \Inte \Omega(u)\left(\frac{1}{|u + (v-y)/R^2|^3} - 1\right) |\phi(y)|^2 dy \right| \\
&\leq
\left(\frac{\omega(2R^{-1})}{(1-2R^{-1})^3} + 2 \|\Omega \|_{L^\infty(\mathbb{S}^2)} R^{-1} \frac{3}{\left(1-2 R^{-1}\right)^4}\right)\int_{\mathbb{R}^3} |\phi|^2 \to 0,
\end{align*}
as $R \to \infty$, which concludes the proof.
\end{proof}

\subsection{Non-existence of minimizers for $0 <\lambda < \lambda_c$}
Assume there is a solution $Q$ to the minimization problem for $0<\lambda < \lambda_c$, then by definition of $\lambda_c(b)$ we have $E(\lambda) = 0$. Moreover, $Q$ satisfies the Euler-Lagrange equation (\ref{eq:EL}) with some chemical potential $\mu < 0$. Then, the same computation as in (\ref{eq:derivee_negative}) shows that $\lambda' \mapsto E(\lambda',b)$ is decreasing in a neighborhood of $\lambda$ contradicting the fact that $\lambda < \lambda_c$

\subsection{Existence of minimizer, for $\lambda = \lambda_c$}

For the existence in the case $\lambda = \lambda_c$, we take a sequence $\lambda_n \to \lambda$ as $n\to\infty$ with $\lambda_n > \lambda$ for all $n\geq 1$ and consider minimizers for the problem with mass constraint $\lambda_n$: for all $n\geq 1$ we take some $\psi_n \in H^1(\mathbb{R}^3)$ with $\int_{\mathbb{R}^3}|\psi_n|^2 = \lambda_n $ and $\mathcal{E}(\psi_n) = E(\lambda_n,b)$. For all $n$, $\psi_n$ verifies the Euler-Lagrange equation (\ref{eq:EL}) for some $\mu_n < 0$. We will use Lemma \ref{lem_rel_compact}, assume first that $\m(\{\psi_n\}) = 0$, then, with the notations of Lemma \ref{lem_virial}, $I(\psi_n) \to 0$ and $Q(\psi_n) \to 0$ from which we deduce, by Lemma \ref{lem_virial}, that $T(\psi_n) \to 0$. But by the Sobolev inequality we have
\begin{equation*}
| I(\psi_n) | \leq C \|\psi_n\|_{L^2(\mathbb{R}^3)} \|\nabla \psi_n\|_{L^2(\mathbb{R}^3)}^{3} = C \lambda_n^{1/2} T(\psi_n)^{3/2},
\end{equation*}
hence $I(\psi_n) = o(T(\psi_n))$ and $\mathcal{E}(\psi_n)\geq 0$. This contradicts the fact that $E(\lambda_n,b) < 0$ for all $n$. This proves that $\m(\{\psi_n\}) >0$, and we can show as before that there is some $\lambda_1 \leq \liminf \{\lambda_n\} = \lambda_c$ such that there is some solution $Q_1$ to the minimization problem with constraint $\lambda_1$. But since there is no minimizer on $[0,\lambda_c)$, we necessarily have $\lambda_1 = \lambda_c$.

\subsection{Upper bound on $\lambda_c(b)$ in the dipolar case}

We now prove the upper bound (\ref{eq:domain_admissible_par_upper_bound}) in the dipolar case. Following the computations done in \cite{Bisset-16}, we use a Gaussian ansatz. We take $n=e_z$ in the definition of $K_{\textrm{dip}}$. For $\lambda, \sigma_\rho,\sigma_z >0$ define
\begin{equation*}
\psi_{\sigma_\rho,\sigma_z}^\lambda =  \sqrt{\frac{8\lambda}{\pi^{3/2} \sigma_\rho^2 \sigma_z}}\, e^{-2\left(\tfrac{\rho^2}{\sigma_\rho^2} + \tfrac{z^2}{\sigma_z^2} \right)},
\end{equation*}
 where $(\rho,z)$ are the cylindrical coordinates. The normalization is such that $\|\psi^\lambda_{\sigma_\rho,\sigma_z}\|_{L^2(\mathbb{R}^3)}^2 = \lambda$. We have
 \begin{equation*}
\frac{\mathcal{E}(\psi_{\sigma_\rho,\sigma_z}^\lambda)}{2\lambda} = \frac{2}{\sigma_\rho^2} + \frac{1}{\sigma_z^2} - \frac{\lambda }{2^{1/2} \pi^{3/2} \sigma_\rho^2 \sigma_z}\left(b f\bigg(\frac{\sigma_{\rho}}{\sigma_z}\bigg) - 1\right) +\frac{ 2^6 \lambda^{3/2} }{5^{5/2} \pi^{9/4} \sigma_\rho^3 \sigma_z^{3/2}} ,
 \end{equation*}
 where 
 \begin{equation*}
 f(x) = \frac{1+2x^2}{1-x^2} - 3x^2 \frac{\tanh^{-1}(\sqrt{1-x^2})}{(1-x^2)^{3/2}}.
 \end{equation*}
 Denoting $\alpha = \sigma_\rho /\sigma_z$ and $Y = \sigma_\rho^{-2}$ we define
 \begin{equation*}
 F_{2}(\lambda,Y,\alpha) = (2+\alpha^2) Y - \lambda Y^{3/2} \frac{\alpha }{2^{1/2} \pi^{3/2}}\left(b f(\alpha) - 1\right) + \lambda^{3/2} Y^{9/4} \alpha^{3/2} \frac{2^6}{5^{5/2} \pi^{9/4}}.
 \end{equation*}
We define
 \begin{equation*}
\lambda_1(\alpha) = \sup \left\{ \lambda > 0 \,  \big| \,  F_{2}(\lambda,Y,\alpha) \geq 0, \forall \, Y \geq 0 \right\}
\end{equation*}
which by the form of $ F_{2}$, exists, is unique and satisfies the following system
 \begin{equation*}
 \left\{
 \begin{array}{ll}
\quad \quad F_{2}(\lambda_1(\alpha),Y_0,\alpha) &= 0, \\
(\partial_Y F_{2})(\lambda_1(\alpha),Y_0,\alpha) &= 0,
 \end{array}
 \right.
\end{equation*}
for some $Y_0 >0$. Solving this system gives
\begin{align*}
\lambda_1(\alpha) = \frac{\pi^{3/2}2^{19/12}}{3^{3/2}} \frac{(2+\alpha^2)^{3/2}}{\alpha \left( b f(\alpha) - 1\right)^{5/2}} .
\end{align*}
Optimizing over $\alpha$ one can obtain numerically
\begin{equation*}
\inf_{\alpha >0} \lambda_1(\alpha) \leq 84.437  \frac{1}{(b-1)^{5/2}}.
\end{equation*}
\appendix
\section{Proof of Lemma \ref{lemma_reg_minimizers}}
	\label{app:Lemma_regularity}

Let $\psi$ be a minimizer of $\mathcal{E}$ for the mass constraint $\lambda >0$. By convexity of the gradient \cite[Theorem 7.8]{LieLos-01} we have $\mathcal{E} (\psi) \geq \mathcal{E}(|\psi|)$. Hence $|\psi|$ also minimizes $\mathcal{E}$ and satisfies the Euler-Lagrange equation (\ref{eq:EL}), therefore, without loss of generality we can assume $\psi \geq 0$.

\paragraph{\textbf{Regularity and positivity}}
Following the proof of \cite[Lemma A.5]{LieSeiYng-00}, for $t>0$, we rewrite (\ref{eq:EL}) as 
\begin{equation}
	\label{eq:EL_Yukawa}
(-\Delta + t^2) \psi = -\left(W - \mu - t^2\right)\psi,
\end{equation}
with $W = ( \psi^2 + b K\star \psi^2 +  \psi^3)$. Solving the equation we find
\begin{equation}\label{Yukawa_resolution}
\psi(x) := - \int_{\mathbb{R}^{3}} Y_t(x-y)\left(W(y) - \mu - t^2\right)\psi(y) dy,
\end{equation}
where $Y_t (x) = (4\pi|x|)^{-1}e^{-t|x|}$ is the Yukawa potential. We want to prove that $-\Delta \psi \in L^2(\mathbb{R}^3)$, by the Euler-Lagrange equation it suffices to show that $\psi \in L^8(\mathbb{R}^3)$. But since $\psi \in H^1(\mathbb{R}^3)$ we have $W \psi \in L^2(\mathbb{R}^3)$ and deduce from the Euler-Lagrange equation that $-\Delta \psi \in L^2(\mathbb{R}^3)$. The usual bootstrapping argument and (\ref{eq:lp_continuity}) give $\psi \in H^{k}(\mathbb{R}^3)$ for all $k\geq 1$ and hence $\psi \in C^{\infty}(\mathbb{R}^3)$. 

In particular we have $-\Delta \psi \in L^{\infty}(\mathbb{R}^3)$, using that $\psi \geq 0$ and applying Harnack's inequality \cite[Theorem 8.20]{GilTru} we deduce that if $\psi$ vanishes then $\psi = 0$ which is excluded by assumption ($\mathcal{E}(\psi) < 0$), hence $\psi >0$.

\medskip

\paragraph{\textbf{Exponential decay}}
Let $0 < t^2 < -\mu /2$. Before proving the exponential decay, we first need to show that $K \ast |\psi|^2$ vanishes at infinity in the sense that
\begin{equation}\label{prop_tends_to_zero}
\forall \varepsilon >0, \,  \exists R_\varepsilon>0 \, \textrm{ s.t. } \sup_{|x|>R_\varepsilon}\{ K\star |\psi|^2 (x)\} < \varepsilon.
\end{equation}
Let $\varepsilon >0$, by the cancellation property (\ref{eq:def_K}) we can write
\begin{gather*}
| K\star \psi^2(x)| = \left| \int_{\mathbb{R}^{3}} \frac{\Omega((x-y)/|x-y|)}{|x-y|^3} \left(\psi(y)^2 -\psi(x)^2\right) dy \right| \\
\leq \int_{B(x,\varepsilon)} \frac{|\Omega((x-y)/|x-y|)|}{|x-y|^{3-\alpha}} \|\psi\|_{L^\infty(\mathbb{R}^3)} \|\psi\|_{C^{0,\alpha}(\mathbb{R}^3)} + 
\int_{B(x,\varepsilon)^c} \frac{|\Omega((x-y)/|x-y|)|}{|x-y|^{3}} \psi(y)^2.
\end{gather*}
Now using that $\Omega\in L^\infty(\mathbb{S}^2)$ and $\psi^2\in L^1(\mathbb{R}^{3})\cap L^\infty(\mathbb{R}^{3})$ we have for all $x\in B(0,R)$, for any $R>0$,
\begin{align}\label{ineq_reg_1}
| K\star \psi^2(x)|
&\leq C\varepsilon \|\Omega\|_{L^\infty(\mathbb{S}^2)} \|\psi\|_{L^\infty(\mathbb{R}^{3})} \|\psi\|_{C^{1}(B(x,1))} + \int_{B(x,\varepsilon)^c} \frac{|\Omega((x-y)/|x-y|)|}{|x-y|^{3}} \psi(y)^2. \nonumber 
\end{align}
Note that the second term above tends to zero at infinity at fixed $\varepsilon$ because, for instance, $\mathds{1}_{|x|>\varepsilon} |x|^{-3}, \psi^2 \in L^2(\mathbb{R}^3)$, this proves that $K\star |\psi|^2$ tends to zero at infinity. Hence $W$ vanishes at $\infty$ and from (\ref{eq:EL_Yukawa}) we deduce that
\begin{equation*}
(-\Delta + t^2) \left(\psi - C Y_t\right) \leq 0, \qquad \textrm{in} \quad B(0,R)^c,
\end{equation*}
for $R$ large enough and $C>0$. Since $\psi$ is bounded, we can find $C>0$ so that $\psi - C Y_t \leq 0$ in $B(0,R)$ (and in particular on $\partial B(0,R)$), applying the maximum principle shows that $\psi \leq C Y_t$, in the whole $\mathbb{R}^3$ and conclude the proof.

\bibliographystyle{siam} %style de biblio
\bibliography{biblio} %nom du fichier biblio.bib (dans le mÍme dossier que le fichier latex qu'on compile)
\end{document}